\documentclass[ij4uq,dvipsnames]{ij4uq}              

\usepackage[hang]{footmisc}
\fancypagestyle{plain}{%
  \fancyhf{}
  \fancyhead[R]{\small {\it International Journal for Uncertainty Quantification}, x(x): \thepage--\pageref{LastPage} (\myyear\today)}
  \fancyfoot[R]{\small\bf\thepage }
  \fancyfoot[L]{\fottitle}
  }
\renewcommand{\dmyy}{19}
\renewcommand{\myyear}{2019}
\renewcommand{\today}{}

\renewcommand{\cfrac}[2]{\ \begin{array}{c}\multicolumn{1}{c|}{#1}\\ \hline\multicolumn{1}{|c}{#2}\end{array}\ }
\newcommand{\specialcell}[2][c]{\begin{tabular}[#1]{@{}c@{}}#2\end{tabular}}
\newtheorem{proposition}{Proposition}[section]
\usepackage[justification=centering]{caption}
\usepackage{xmpmulti}
\usepackage{booktabs, arydshln}
\usepackage{bbm}
\usepackage{graphicx, tikz}
\usetikzlibrary{decorations.pathmorphing,decorations.pathreplacing,decorations.shapes,mindmap,backgrounds}

\begin{document}

\volume{Volume x, Issue x, \myyear\today}
\title{Optimal Uncertainty Quantification of a Risk Measurement From a Thermal-Hydraulic Code using Canonical Moments}
\titlehead{OUQ using Canonical Moments}
\authorhead{J. Stenger, F. Bamboa, M. Keller, \& B. Iooss}
\corrauthor[1,2]{J. Stenger}
\author[2]{F. Gamboa}
\author[1]{M. Keller}
\author[1,2]{B. Iooss}
\corremail{jerome.stenger@edf.fr}
\corraddress{EDF R$\&$D, 6 quai Watier, 78401 Chatou, France}
\address[1]{EDF R$\&$D, 6 quai Watier, 78401 Chatou, France}
\address[2]{Universit\'e Paul Sabatier, 118 route de Narbonne, 31400 Toulouse, France}

\dataO{12/04/2019}
\dataF{12/04/2019}

\abstract{In uncertainty quantification studies, a major topic of interest is to assess the uncertainties tainting the results of a computer simulation. In this work we gain robustness on the quantification of a risk measurement by accounting for all sources of uncertainties tainting the inputs of a computer code. To that extent, we evaluate the maximum quantile over a class of bounded distributions satisfying constraints on their moments. Two options are available when dealing with such complex optimization problems: one can either optimize under constraints; or preferably, one should reformulate the objective function. We identify a well suited parameterization to compute the maximal quantile based on the theory of canonical moments. It allows an effective, free of constraints, optimization. This methodology is applied on an industrial computer code related to nuclear safety.
}


\maketitle

\section{Introduction}

Computer codes are increasingly used to measure safety margins in nuclear accident management analysis instead of conservative procedures \citep{pourgol-mohamad_integrated_2009}. In this context, it is essential to evaluate the accuracy of the numerical model results, whose uncertainties come mainly from the lack of knowledge of the underlying physic and the model input parameters. The Best Estimate Plus Uncertainty (BEPU) methods \citep{iooss_advanced_2018} were developed in safety analyses, especially for the large break loss of coolant accident (see \cite{prosek_state---art_2007}, \cite{sanchez-saez_uncertainty_2018}). Its principles rely mainly on a probabilistic modeling of the model input uncertainties, on Monte Carlo sampling for running the thermal-hydraulic computer code on sets of input, and on the application of statistical tools to infer high quantiles of the scalar output variables of interest \citep{wallis_uncertainty_2004} . 

This takes place in a more general setting, known as Uncertainty Quantification (UQ) methods \citep{de_rocquigny_uncertainty_2008}. Quantitative assessment of the uncertainties tainting the results of computer simulations is a major topic of interest in both industrial and scientific communities. In the context of nuclear safety, the computer models are expensive to run. Uncertainty propagation, risk measurement such as high quantile inference, or system robustness analysis become a difficult task using such models. In order to circumvent this problem, a widely accepted method consists in replacing the cpu time expensive numerical simulations by inexpensive mathematical functions called metamodels. This metamodel is build from a set of computer code simulations that must be as representative as possible of the code in the variation domain of its uncertain inputs. Generally, space-filling designs of experiments are created with a given budget of $n$ code evaluations, that provide a full coverage of the input space \citep{fang_design_2005}. In the presence of a high number of input parameters, screening strategies are then performed in order to identify the Primary Influential Inputs (PII) on the model output variability and rank them by decreasing influence. From the learning sample, a metamodel is therefore built to fit the simulator output, considering only the PII as the explanatory inputs, while the remaining inputs remain fixed to a default value. Among all the metamodel-based solutions (polynomials, splines, neural networks, etc.), Gaussian process metamodeling, also known as Kriging \citep{rasmussen_gaussian_2005}, has been very attractive. It makes the assumption that the response is a realization of a Gaussian process, conditioned on code observations. This approach provides the basis for statistical inference. In that, we dispose of simple analytical formulations of the predictor and the mean squared error of the predictions. The metamodel is then validated before being used. 

Once the predictive metamodel has been built, it can be used to perform uncertainty propagation and in particular, estimate probabilities or, as here, quantiles  (see \cite{cannamela_controlled_2008}, \cite{lorenzo_assessment_2011}). This measure of risk will be designated from now on as our quantity of interest. The most trivial approach to estimate a quantile with a Gp metamodel, called plug-in approach, is to apply the quantile definition to the predictor of the metamodel. As the expectation of the Gp mean is a deterministic function of the input, this provides a deterministic expression of the quantile but no confidence intervals. Moreover for high quantiles, this methods tends to substantially underestimate the true quantile \cite{cannamela_controlled_2008}. To assess this problem,
\cite{oakley_estimating_2004} has proposed to take into account the covariance structure of the Gp metamodel. The quantile definition is therefore applied to the global Gp metamodel and yields a random variable, whose expectation can be considered as the quantile estimator and its variance an indicator of the accuracy of its prediction. This full-Gp approach leads to confidence intervals. In practice, the estimation of a quantile with the full Gp approach is based on stochastic simulations (conditional on the learning sample) of the Gp metamodel.

This overall methodology yields the estimation of the $p$-quantile of the model output. In nuclear safety, as in other engineering domains, methods of conservative computation of quantiles have been largely studied \cite{hessling_robustness_2015}. Though the above construction work largely increases the robustness of the metamodel, the evaluation of the quantile remains tainted by the uncertainty of the input distributions. For simplicity, the inputs probability densities are usually chosen in parametric families (uniform, normal, log-normal, etc. See for instance Table \ref{tab: Constraints for CATHARE model}). The distribution's parameters are itself set with the available information coming from data and/or an expert opinion. This information is often reduced to an input’s mean value or a variance. Nevertheless, the distribution choice differs inevitably from reality. This uncertainty on the input probability densities is propagated to the quantile, hence, different choices of distributions will lead to different quantile values, thus different safety margins. 

In this work, we propose to gain robustness on the quantification of this measure of risk. We aim to account for the uncertainty on the input distributions by evaluating the maximum quantile over a class of probability measures $\mathcal{A}$. In this optimization problem, the set $\mathcal{A}$ must be large enough to effectively represent our uncertainty on the inputs, but not too large in order to keep the estimation of the quantile representative of the physical phenomena. For example, the maximum quantile over the very large class $\mathcal{A}=\left\lbrace \text{all distributions} \right\rbrace$, proposed in \citep{huber_use_1973}, will certainly be too conservative to remain physically meaningful. Several articles which discuss possible choices of classes of distributions can be found in the literature of Bayesian robustness (see \cite{ruggeri_robust_2005}). \cite{deroberts_bayesian_1981}, consider a class of measures specified by a type of upper and lower envelope on their density. \cite{sivaganesan_ranges_1989} study the class of unimodal distributions. In more recent work, \cite{owhadi_optimal_2013} propose to optimize the measure of risk over a class of distributions specified by constraints on their \textit{generalized} moments. Their work (and ours) is called Optimal Uncertainty Quantification (OUQ), because given a set of assumptions and information, there exist optimal bounds on uncertainties. In practical engineering cases, the available information on an input distribution is often reduced to the knowledge of its mean and/or variance. This is why in this paper, we are interested in a specific case of the framework introduced by \cite{owhadi_optimal_2013}. We consider the class of measures known by some of their \textit{classical} moments, which we refer to as the moment class:
\begin{eqnarray}
	\mathcal{A} & = & \bigg\{ \mu = \otimes \mu_i \in \bigotimes_{i=1}^{d} \mathcal{P} ([l_i,u_i])\; | \; \mathbb{E}_{\mu_i}[x^j] = c_{i}^{(j)} \ ,  \label{eq: Optimization set} \\ & & \quad  c_{i}^{(j)}\in\mathbb{R}, \text{ for } 1\leq j\leq N_i \text{ and } 1\leq i\leq  d \bigg\}\ , \nonumber
\end{eqnarray}
where $\mathcal{P} ([l_i,u_i])$ denotes the set of scalar probability measure on the interval $[l_i, u_i]$. This set traduces simply that each random input is restricted to a given range and has some of its moments fixed. The tensorial product of measure sets traduces the mutual independence of the $d$ inputs. The choice of range enforcement is not a very strong hypothesis, as input variables represent physical parameters which are rarely unbounded.

The solution to our optimization problem is numerically computed thanks to the OUQ reduction theorem (\cite{owhadi_optimal_2013}, \cite{winkler_extreme_1988}). This theorem states that the measure corresponding to the extremal CDF (can be extended to quantile through Proposition \ref{THM : DUALITY THEOREM}), is located on the extreme points of the distribution set. In the context of the moment class, the extreme distributions are located on the $d$-fold product of finite convex combinations of Dirac masses:
\begin{eqnarray}
	\mathcal{A}_\Delta & = & \left\lbrace \mu\in\mathcal{A} \; | \; \mu_i = \sum_{k=1}^{N_i+1} w_{i}^{(k)} \delta_{x_{i}^{(k)}}\ \text{ for }\  1\leq i\leq d\right\rbrace\ , \label{eq: Discrete Optimization set}
\end{eqnarray}
To be more specific it holds that when $N$ pieces of information are available on the moments of a scalar measure $\mu$, it is enough to pretend that the measure is supported on at most $N+1$ points. This powerful theorem gives the basis for practical optimization of our optimal quantity of interest. In this matter, Semi-Definite-Programming \citep{henrion_gloptipoly_2009} has been already already explored by \cite{betro_robust_2005} and \cite{lasserre_moments_2010}, but the deterministic solver used rapidly reaches its limitation as the dimension of the problem increases. One can also find in the literature a Python toolbox developed by \cite{mckerns_optimal_2012} called Mystic framework that fully integrates the OUQ framework. However, it was built as a generic tool for generalized moment problems and the enforcement of the moment constraints is not optimal. By restricting the work to our moment class, we propose an original and practical approach based on the theory of canonical moments \citep{dette_theory_1997}. Canonical moments of a measure can be seen as the relative position of its moment sequence in the moment space. It is inherent to the measure and therefore presents many interesting properties. It allows to explore very efficiently the optimization space $\mathcal{A}_{\Delta}$, were the maximum quantile is to be found. Hence, we rewrite the optimization problem on the highly constrained set $\mathcal{A}_{\Delta}$ into a simplified and constraints free optimization problem.

The paper proceeds as follows. Section \ref{sec: OUQ principles} describes the OUQ framework and the OUQ reduction theorem. In Section \ref{sec: Methodology}, we then describe step by step the algorithm calculating our quantity of interest with the canonical moments parameterization. We present in Section \ref{sec: Modified algorithm}, an extended algorithm to deal with inequality constraints on the moments. Section \ref{sec: Numerical tests on toy example} and \ref{sec: Real case study} are dedicated to the application of our algorithm to a toy example, and to the peak cladding temperature for the IBLOCA application presented in the introduction. Section \ref{sec: Conclusion} gives some conclusions and perspectives.

\section{OUQ principles} 
\label{sec: OUQ principles}

\subsection{Duality transformation} 
In this work, we consider the quantile of the output of a computer code $G:\mathbb{R}^d\rightarrow \mathbb{R}$, seen as a black box function. As we said, in order to gain robustness on the risk measurement, our goal is to find the maximum quantile over the moment class $\mathcal{A}$ described in Equation (\ref{eq: Optimization set}). The objective value writes:
\begin{eqnarray}
    \overline{Q}_\mathcal{A}(p) & = & \sup_{\mu \in \mathcal{A}} \bigg[ \inf \left\lbrace h  \in \mathbb{R}\ | \ F_{\mu}(h) \geq p \right\rbrace \bigg] \label{eq : Objective value} \ , \\
                                & = & \sup_{\mu \in \mathcal{A}} \bigg[ \inf \left\lbrace h  \in \mathbb{R}\ | \ \mathbb{P}_{\mu} (G(X) \leq h) \geq p \right\rbrace \bigg] \ , \nonumber
\end{eqnarray}
The objective value written as a quantile \eqref{eq : Objective value} is not very convenient to work with. In order to applied the OUQ reduction Theorem \ref{th: OUQ reduction theorem} \citep{owhadi_optimal_2013}, one must optimize an affine functional of the measure. In particular, it is necessary to optimize a probability instead of a quantile. The following result, illustrated in Figure \ref{fig : Duality transformation}, can be interpreted as a duality transformation of our optimization problem \eqref{eq : Objective value}, into the optimization of a probability of failure (\textsc{p.o.f}). The proof is postponed to \ref{app: proof}. 
\begin{proposition} The following duality result holds
    \[ \overline{Q}_\mathcal{A}(p) = \inf \left\lbrace h  \in \mathbb{R}\; | \; \inf_{\mu\in\mathcal{A}} F_{\mu}(h) \geq p \right\rbrace\ . \]
    \label{THM : DUALITY THEOREM}
\end{proposition}

\begin{figure}[htb]
    \centering
     \begin{tikzpicture}[xscale=2.5, yscale=3.5]
        \draw[->] (0,0) -- (3.2,0);
        \draw[->] (0,0) -- (0,1.05);
        \node[below left] at (0,0) {$0$};
        \node[left] at (0,1) {$1$};
        \node[left] at (0,0.4) {$F_\mu (h)$};
        \node[below right] at (3.2,0) {$h$};
        \draw[domain=0:1.92, color=black, samples=200] plot(\x,{min(0.829*sqrt(sqrt(\x)) - 0.05*sin(deg(3*\x)),1)});
        \draw[color=black] (1.92,1) -- (3.2,1);
        \draw[domain=0:3.2, color=black, samples=100] plot(\x,{min(0.7*ln(\x+1)+0.05*sin(deg(3*\x)),1)});
        \draw[domain=0:3.2, color=black] plot(\x,{0.04*exp(\x)-0.04});
        \draw[dashed] (0,0.75) -- (2.98,0.75);
        \draw[dashed] (0.8,0) -- (0.8,0.75);		
        \draw[dashed] (1.98,0) -- (1.98,0.75);
        \draw[dashed] (2.98,0) -- (2.98,0.75);
        \node[left] at (0,0.75) {$p$};
        \node[below, color=black] at (0.8,0) {$x_2$};
        \node[below, color=black] at (1.98,0) {$x_1$};
        \node[below, color=black] at (2.98,0) {$x_{max}$};
        \draw (2.6,0.2) node[below] {$\inf F_{\mu}(h)$} to[out=100, in=-10] (2.385,0.365);
    \end{tikzpicture}
    \caption{Illustration of the duality result \ref{THM : DUALITY THEOREM}. The lower curve represents the CDF lower envelope; we can see that the maximum quantile $x_{max}$ is actually the quantile of the lowest CDF.}
	\label{fig : Duality transformation}
\end{figure}
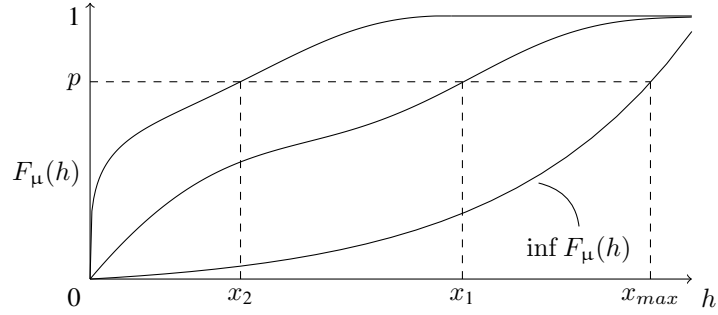

Proposition \ref{THM : DUALITY THEOREM} reads like this: the optimal quantile over a class of distributions is equal to the quantile of the CDF lower envelope.
Our problem is therefore to evaluate the lowest probability of failure $\inf_{\mu \in \mathcal{A}} \mathbb{P}_{\mu}(G(X)\leq h)$ for a given threshold $h$.

\subsection{Reduction Theorem}
Under the form of Proposition \ref{THM : DUALITY THEOREM}, the OUQ reduction theorem applies (see \cite{owhadi_optimal_2013}, \cite{winkler_extreme_1988}). It states that the optimal solution of the \textsc{p.o.f} optimization is a product of discrete measures. A general form of the theorem reads as follows:

\begin{theorem}[OUQ reduction {\cite[p.37]{owhadi_optimal_2013}}]
	Suppose that $\mathcal{X} := \mathcal{X}_1 \times \dots \times \mathcal{X}_d$ is a product of Radon spaces. Let 
	\[\mathcal{A} := \left\lbrace (G,\mu)\ \begin{array}{|l}
		G:\mathcal{X}\rightarrow \mathcal{Y} \text{, is a real valued measurable function}, \\
		\mu = \mu_1 \otimes \dots \otimes \mu_p \in \bigotimes_{i=1}^d \mathcal{P}(\mathcal{X}_i) \ ,\\
		\text{for each G, and for some measurable functions }\\ \varphi_l:\mathcal{X}\rightarrow \mathbb{R} \text{ and } \varphi_i^{(j)}:\mathcal{X}_i \rightarrow \mathbb{R}\ , \\
		\qquad \qquad \bullet\ \mathbb{E}_{\mu} [\varphi_l] \leq 0 \text{ for } l=1,\dots,N_0\ ,  \\
		\qquad \qquad \bullet\ \mathbb{E}_{\mu_i}[\varphi_i^{(j)}] \leq 0 \text{ for } j=1, \dots, N_i \text{ and } i=1,\dots, d 
	\end{array}\right\rbrace\]
	Let $\Delta_n(\mathcal{X})$ be the set of all discrete measure supported on at most $n+1$ points of $\mathcal{X}$, and
	\[\mathcal{A}_{\Delta} := \left\lbrace (G,\mu) \in \mathcal{A} \ |\ \mu_i \in \Delta_{N_0+N_i}(\mathcal{X}_i) \right\rbrace \ .\]
	Let $q$ be a measurable real function on $\mathcal{X}\times \mathcal{Y}$. Then
	\[\displaystyle \sup_{(G,\mu)\in\mathcal{A}} \mathbb{E}_{\mu}[q(X,G(X))] = \sup_{(G,\mu)\in\mathcal{A}_{\Delta}} \mathbb{E}_{\mu}[q(X,G(X))]\ .\]
	\label{th: OUQ reduction theorem}
\end{theorem}
This theorem derives from the work of \cite{winkler_extreme_1988}, who has shown that the extreme measures of a moment class $\{ \mu \in \mathcal{P}(\mathcal{X}) \ | \ \mathbb{E}_\mu [\varphi_1] \leq 0 , \allowbreak \dots , \mathbb{E}_\mu [\varphi_n] \allowbreak \leq 0 \}$ are the discrete measures that are supported on at most $n+1$ points. The strength of Theorem \ref{th: OUQ reduction theorem} is that it extends the result to a tensorial product of moment sets. The proof relies on a recursive argument using Winkler's classification on every set $\mathcal{X}_i$. A remarkable fact is that, as long as the quantity to be optimized is an affine function of the underlying measure $\mu$, this theorem remains true whatever the function $G$ and the quantity of interest $q$ are.
Applying Theorem \ref{th: OUQ reduction theorem} to the optimization of the probability of failure, it is rewritten as:
\begin{eqnarray}
\inf_{\mu \in \mathcal{A}} F_{\mu}(h)\ & = & \inf_{\mu \in \mathcal{A}_\Delta} F_{\mu}(h)\ , \nonumber \\
 & = & \inf_{\mu \in \mathcal{A}_\Delta} \mathbb{P}_{\mu} (G(X)\leq h) \ , \nonumber \\ & = & \inf_{\mu \in \mathcal{A}_\Delta} \sum_{i_1=1}^{N_1+1} \dots \sum_{i_d=1}^{N_d+1} \omega_{1}^{(i_1)} \dots \omega_{d}^{(i_d)}\ \mathbbm{1}_{\{G(x_{1}^{(i_1)}, \dots, x_{d}^{(i_d)}) \leq h\}} \ ,
\label{eq: Probability of faillure sum of weights}
\end{eqnarray}

\section{Parameterization using canonical moments}
\label{sec: Methodology}
The optimization problem in Equation (\ref{eq: Probability of faillure sum of weights}) shows that the weights and positions of the input distributions provide a natural parameterization for the computation of the \textsc{p.o.f}. However, in order to compute the lowest \textsc{p.o.f}, one must be able to explore the whole set of admissible measures $\mathcal{A}_{\Delta}$. Two ways to handle the problem appear. The first one consists in optimizing the objective value $F_{\mu}(h)$ under constraints, that is $\mu \in \mathcal{A}_{\Delta}$: this is the approach taken by \cite{mckerns_optimal_2012} with the Mystic Framework. The second option, always favored when possible, consists in reformulating the objective function. This requires to identify a new parameterization adapted to the problem. Here, canonical moments \citep{dette_theory_1997} provide a surprisingly well tailored reparameterization.

The work on canonical moments was first introduced by \cite{skibinsky_range_1967}. His main contribution covered the original study of the geometric aspect of general moment space \citep{skibinsky_maximum_1977}, \citep{skibinsky_principal_1986}. In a number of further papers, Skibinsky proves numerous other interesting properties of the canonical moments.  \cite{dette_theory_1997} have shown the intrinsic relation between a measure $\mu$ and its canonical moments. They highlight the interest of canonical moments in many areas of statistics, probability and analysis such as problem of design of experiments, or the Hausdorff moment problem \citep{hausdorff_momentprobleme_1923}. In the following, we describe step by step the algorithm used to transform the optimization problem of Equation \ref{eq: Probability of faillure sum of weights} under the canonical moments parameterization.

\subsection{Step 1. From classical moments to canonical moments}
\label{subsec: Step 1}
We enforce some moments on the input distributions of the code $G$. In this section, we present how to transform these \textit{classical} moment constraints, into canonical moments constraints.  

We define the moment space $M := M(a,b) = \{ \mathbf{c}(\mu) \ | \ \mu\in\mathcal{P}([a,b])\}$ where $\mathbf{c}(\mu)$ denote the sequence of all moments of some measure $\mu$. The $n$th moment space $M_n$ is defined by projecting $M$ onto its first $n$ coordinates, $M_n = \{ \mathbf{c}_n(\mu) = (c_1,\dots, c_n)\; | \; \mu\in\mathcal{P}([a,b])\}$. $M_2$ is depicted in Figure \ref{fig: Moment set M2}. We first define the extreme values, 
\begin{align*}
    & c_{n+1}^{+} = \max \left\lbrace c\in \mathbb{R} : (c_1,\dots, c_n,c) \in M_{n+1} \right\rbrace \ ,\\
    & c_{n+1}^{-} = \min \left\lbrace c\in \mathbb{R} : (c_1,\dots, c_n,c) \in M_{n+1} \right\rbrace \ ,
\end{align*}
which represent the maximum and minimum values of the $(n+1)$th moment that a measure can have, when its moments up to order $n$ equal to $\mathbf{c}_n$ . The $n$th canonical moment is then defined recursively as 
\begin{equation}
	p_n=p_n(\mathbf{c})=\frac{c_n-c_n^{-}}{c_n^{+}-c_n^{-}}\ .
	\label{eq: Canonical moment definition}
\end{equation}

\begin{figure}
    \centering        
    \begin{tikzpicture}[scale=4]
        \draw (0,0) rectangle (1,1);
        \draw[domain=0:1, fill=gray!30!white, opacity=0.4] plot(\x,{\x*\x});
        \draw[opacity=0.4] (0,0) -- (1,1);
        \node[left] at (0,1) {$1$};
        \node[below left] at (0,0) {$0$};
        \node[below] at (1,0) {$1$};
        \node at (0.77,0.7) {$M_2$};
        \node[below] at (0.6,0) {$c_1$};
        \node[left] at (0,0.6) {$c_2^+ = c_1$};
        \node[left] at (0,0.36) {$c_2^- = c_1^2$};
        \draw[dashed] (0,0.6) -- (0.6,0.6) -- (0.6,0);
        \draw[dashed] (0,0.36) -- (0.6,0.36);
    \end{tikzpicture}
    \caption{The moment set $M_2$ and definition of $c_2^+$ and $c_2^-$ for $(a,b)=(0,1)$.}
	\label{fig: Moment set M2}
\end{figure}
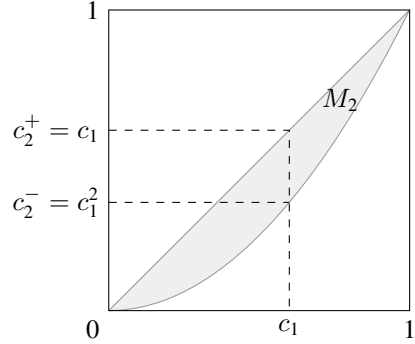

Note that the canonical moments are defined up to the degree $N = N(\mathbf{c}) = \min \left\lbrace \right. n \in \mathbb{N} \ |\ \mathbf{c}_n  \in \partial  M_n \left. \right\rbrace $, and $p_N$ is either $0$ or $1$. Indeed, we know from \citep[Theorem 1.2.5]{dette_theory_1997} that $\mathbf{c}_n \in \partial M_n$ implies that the underlying $\mu$ is uniquely determined, so that, $c_n^+ = c_n^-$. We also introduce the quantity $\zeta_n = (1-p_{n-1})p_n$ that will be of some importance in the following. The very nice properties of canonical moments are that they belong to $[0,1]$ and are invariant by any affine transformation of the support of the underlying measures. Hence, we may restrict ourselves to the case $a=0$, $b=1$. 

Therefore, for every $i=1,\dots, d$, the support of the measure $\mu_i$ is transformed into $[0,1]$ using the affine transformation $y = l_i +(u_i - l_i)x$. The sequences of moments of the corresponding measures are written $\mathbf{c}'_i = ({c'}_{i}^{(1)}, \dots, {c'}_{i}^{(N_i)})$ where ${c'}_{i}^{(j)}$ reads
\begin{equation}
{c'}_{i}^{(j)} = \frac{1}{(u_i- l_i)^{j}} \sum_{k=0}^{j} \binom{j}{k} (- l_i)^{j -k} c_{i}^{(k)} \ , \quad \text{ for } j = 1,\dots, N_i\, .
\label{eq: Moment's affine transformation}
\end{equation} 
Given a sequence of moment constraints $(c_{i}^{(j)})_{1\leq j \leq N_i}$ enforced to the $i$th input, it is then possible to calculate the corresponding sequence of canonical moments $\mathbf{p}_i = (p_{i}^{(1)} , \dots, p_{i}^{(N_i)})$. \cite[p. 29]{dette_theory_1997} propose a recursive algorithm named \textit{Q-D algorithm} that allows this computation. It drastically fastens the computational time compared to the raw formula that consists of computing Hankel determinants \citep[p. 32]{dette_theory_1997}.
In practical applications, we wish to enforce low order of moments, typically order 2 (see for instance Table \ref{tab:initial distribution hydraulic model}). In this case we dispose of the simple analytical formulas 
\begin{eqnarray*}
    p_1 = c_1 & \quad ,\quad & p_2=\frac{c_2-c_1^2}{c_1(1-c_1)} \ .
\end{eqnarray*} 
One can easily see that enforcing $N_i$ moments or $N_i$ canonical moments to the $i$th input is equivalent. Indeed, Equations  (\ref{eq: Canonical moment definition}) and (\ref{eq: Moment's affine transformation}) can be inverted.

\subsection{Step 2. From canonical moments to support points}
\label{subsec: Step 2}
From a given sequence of canonical moments, one wishes to reconstruct the support of a discrete measure. 
We introduce the Stieltjes Transform, which connects canonical moments of a measure to its support. The Stieltjes transform \citep{dette_theory_1997} of a scalar measure $\mu$ is defined as 
\[S(z)=S(z,\mu) = \int_a^b \frac{d\mu (x)}{z-x}\ , \quad (z\in\mathbb{C}\backslash\{\text{supp}(\mu)\}) \ .\]
The transform $S(z, \mu)$ is an analytic function of $z$ in $\mathbb{C} \backslash \text{supp}(\mu)$. If $\mu$ has a finite support then 
\[S(z)= \int_a^b \frac{d\mu (x)}{z-x}= \sum_{i=1}^n \frac{\omega_i}{z-x_i}\ , \]
where the support points of the scalar measure $\mu$ are distinct and denoted by $x_1,\allowbreak \dots, x_n$, with corresponding weights $\omega_1,\dots, \omega_n$. Alternatively, the weights are given by $\omega_i = \lim_{z \rightarrow x_i} (z-x_i)S(z)$. We can rewrite the transform as a ratio of two polynomials with no common zeros. The zeros of the denominator being the support of $\mu$. 
\begin{equation}
S(z) = \frac{Q^{(n-1)}(z)}{P_*^{(n)} (z)} \ ,
\label{eq : Stieljes transform Q/P}
\end{equation} 
where $P_*^{(n)} (z)= \prod_{i=1}^n (z - x_i)$ and 
\[ \omega_i = \frac{Q^{(n-1)}(x_i)}{\frac{d}{dx} P_*^{(n)} (x) |_{x=x_i}}\ . \]

The Stieljes transform can also be written as a continuous fraction, some basic definitions and properties of continuous fraction are postponed to \ref{app: continous fraction}. 

\begin{theorem}[{\cite[Theorem 3.3.1]{dette_theory_1997}}]
Let $\mu$ be a probability measure on the interval $[a,b]$ and $z\in \mathbb{C} \backslash [a,b]$, then the Stieltjes transform of $\mu$ has the continued fraction expansion (see Appendix \ref{app: continous fraction} for notation):
\begin{align*}
S(z) & = \cfrac{1}{z-a} \ -\  \cfrac{\zeta_1 (b-a)}{1} \ -\  \cfrac{\zeta_2 (b-a)}{z-a} \ -\  \dots \label{eq : Stieljes transform}\\
& = \cfrac{1}{z-a- \zeta_1 (b-a)} \ -\  \cfrac{\zeta_1 \zeta_2 (b-a)^2}{z-a - (\zeta_2+ \zeta_3)(b-a)} \\ &  -\  \cfrac{\zeta_3 \zeta_4 (b-a)^2}{z-a- (\zeta_4 +\zeta_5) (b-a)} \ -\  \dots \nonumber
\end{align*} 
\label{Th: Stieljes transform}
Where we recall that $\zeta_n := p_{n-1}(1-p_n)$.
\end{theorem} 

Theorem \ref{Th: Stieljes transform} states that the Stieltjes transform can be computed when one knows the canonicals moments. It immediately follows from Equation (\ref{eq : Stieljes transform Q/P}), Theorem \ref{Th: Stieljes transform} and Lemma \ref{Lem : Continued Fraction} that we have the following recursive formula for $P_\ast^{(n)}$
\begin{equation}
 P_{\ast}^{(k+1)} (x)=(x-a-(b-a)(\zeta_{2k}+\zeta_{2k+1}))P_{\ast}^{(k)}(x)-(b-a)^2 \zeta_{2k-1} \zeta_{2k} P_{\ast}^{(k-1)}(x) \ ,
\label{eq : Recursive formula polynomial}
\end{equation}
where $P_{\ast}^{(-1)}=0$, $P_\ast^{(0)}=1$. The support of $\mu$ thus consists of the roots of $P_\ast^{(n)}$. This obviously leads to the following theorem.

\begin{theorem}[{\cite[Theorem 3.6.1]{dette_theory_1997}}]
Let $\mu$ denote a measure on the interval $[a, b]$ supported on $n$ points with canonical moments $p_1 , p_2 , \dots$. Then, the support of $\mu$ is the set of $\{x:P_\ast^{(n)}(x)=0\}$ defined by Equation (\ref{eq : Recursive formula polynomial}).
\label{th : Dette Generation of measure}
\end{theorem}

In the following we consider a fixed sequence of moments $\mathbf{c}_n = (c_1, \dots,\allowbreak c_n)\allowbreak \in M_n$, let $\mu$ be a measure supported on at most $n+1$ points, such that its moments up to order $n$ coincide to $\mathbf{p}_n=(p_1, \dots, p_n)$ the corresponding sequence of canonical moments related to $\mathbf{c}_n$, as described in Section \ref{subsec: Step 1}. 
Corollary \ref{Th : Generation of measure} is the moment version of Theorem \ref{th : Dette Generation of measure}. 
The only difficulty compared to Theorem \ref{th : Dette Generation of measure} is that one tries to generate admissible measures supported on \textit{at most} $n+1$ Dirac masses. Given a measure supported on strictly less than $n+1$ points, the question is therefore to know whether it makes sense to evaluate the $n+1$ roots of $P_{\ast}^{(n+1)}$. A limit argument is used for the proof.

\begin{corollary}
    Consider a sequence of moments $\mathbf{c}_n = (c_1,\allowbreak \dots,\allowbreak c_n) \in M_n$, and the set of measures 
    \[ \mathcal{A}_\Delta = \left\lbrace \mu = \sum_{i=1}^{n+1} \omega_i \delta_{x_i} \in \mathcal{P}([a,b]) \ | \ \mathbbm{E}_{\mu}(x^j) = c_j, \ j=1,\dots,n \right\rbrace \ .\]
    We define 
    \[ \Gamma=\left\lbrace (p_{n+1}, \dots, p_{2n+1})\in [0,1]^{n+1} \ | \ p_i\in \{0,1\} \Rightarrow p_k=0 ,\ k > i \right\rbrace\ . \]
	Then there exists a bijection between $\mathcal{A}_{\Delta}$ and $\Gamma$.
\label{Th : Generation of measure}
\begin{figure}[htb]
	\centering
    \begin{tikzpicture}[remember picture, scale=3]
		\node (A) at (-0.7,0) {$\mu\in\mathcal{A}_\Delta$};
		\node (B) at (0.7,0) {$P_{\ast}^{(n+1)}$};
		\node (C) at (0,-0.8) {$\Gamma$};
		\draw[->] (A) to node[below, rotate=-50, text width=3.5cm, scale=0.75]{canonical moments $(p_{n+1},\dots, p_{2n+1})$} (C) ;
		\draw[->] (C) to node[below, rotate=50, scale=0.75]{{iterative formula}} (B);
		\draw[->] (B) to node[above, scale=0.75]{{roots provide support}} (A);
	\end{tikzpicture}
	\caption{Relation between the set of admissible measures and the canonical moments.}
\end{figure}
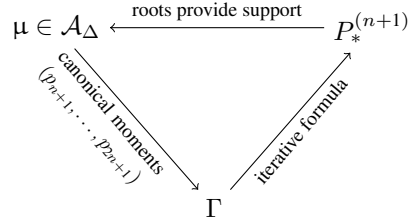
\end{corollary}

\begin{proof}
Without loss of generality we can always assume $a=0$ and $b=1$ as the problem is invariant using affine transformation. We first consider the case where $\text{card}(\text{supp}(\mu))$ is exactly $n+1$. From Theorem \ref{th : Dette Generation of measure}, the polynomial $P_{\ast}^{(+1)}$ is well defined with $n+1$ distinct roots corresponding to the support of $\mu$. Notices that this implies that $(p_1, \dots, p_{2n-1})$ belongs to $]0,1[^{2n-1}$ and that $p_{2n}, p_{2n+1}$ or $p_{2n+2}$ belong to $\{ 0,1 \}$. 

Now, the functions $g(x,z) = 1/(z-x)$ are equicontinuous for $z$ in any compact region which has a positive distance from $[0,1]$. The Stieljes transform is a finite sum of equicontinuous functions and therefore also equicontinuous. Thus if a measure $\mu$ converges weakly to $\mu^*$, the convergence must be uniform in any compact set with positive distance from $[0,1]$ (see \cite{royden_real_1968}). It is then always possible to restrict ourselves to measures of cardinal $m < n+1$, by letting $p_k$ converge to $0$ or $1$ for $2m-2\leq k \leq 2m $. Note that by doing so the polynomials $P_\ast^{(m)}$ and $P_{\ast}^{(n+1)}$ will have the same roots. But, $P_{\ast}^{(n+1)}$ and $Q^{(n-1)}$ will have some others roots of multiplicity strictly equal (see Equation (\ref{eq : Stieljes transform Q/P}) and (\ref{eq : Recursive formula polynomial})). The corresponding weights of these roots are vanishing, so that the measures extracted from $P_\ast^{(m)}$ and $P_{\ast}^{(n+1)}$ are the same. 
\end{proof}
\begin{remark}
From a computational point of view, as the proof relies on a limit argument, we can always generate $p_k \in ]0,1[\ ,\  \text{ for }  n+1\leq k \leq 2n+1 $. This prevents the condition $p_k \in \{ 0,1\} \Rightarrow p_j = 0$ for $j>k$. 
\end{remark}

We use Section \ref{subsec: Step 1} to transform the $N_i$ constraints on the moments of the $i$th input into $N_i$ canonical moment constraints. However, the construction of the polynomial $P_{i\ast}^{(N_i+1)}$ requires the sequence $(p_1^{(i)}, \dots, p_{2N_i +1}^{(i)})$. The $N_i$ first canonical moments of this sequence are known by the constraints, while the canonical moments $(p_k)_{N_i+1\leq k \leq 2N_i+1}\in\Gamma$ constitute $N_i+1$ free parameters, in equal number to the cardinal of $\mu_i$. The computation of $\Gamma$ is very simple, it is basically done by random generation of $N_i+1$ numbers in $]0,1[$, yet it allows to generate the support of all the measures in $\mathcal{A}_{\Delta}$. This provides a very nice parameterization of the problem that takes naturally into account the constraints. 

We also provide a geometrical interpretation of this results. Consider the support of point $(x_1,\dots,x_{n+1})$ of a discrete measure in $\mathcal{A}_\Delta$. The set of support point for all measure in $\mathcal{A}_\Delta$ is a manifold in $\mathbb{R}^{n+1}$. Corollary \ref{Th : Generation of measure} described the structure of the manifold and states it is an algebraic variety. Indeed, it is the zero locus of the set of polynomial $P_\ast^{(n)}$ built from a sequence of canonical moments in $(p_1,\dots,p_n)+\Gamma$.

\subsection{Step 3. From support points to weights}
\label{subsec: Step 3}
From the positions of a discrete measure $\mu$ in $\mathcal{A}_{\Delta}$ generated in Section \ref{subsec: Step 2}. We easily recover the associated weights. Indeed, we enforce $N_i$ constraints on the moments of a scalar measure $\mu_i = \sum_{j=1}^{N_i+1} \omega_j^{(i)} \delta_{x_j^{(i)}}$, supported by at most $N_i+1$ points according to Theorem \ref{th: OUQ reduction theorem}. A noticeable fact is that as soon as the $N_i+1$ support points of the distribution are set, the corresponding weights are uniquely determined. Indeed, the $N_i$ constraints lead to $N_i$ equations, and one last equation derives from the measure mass equals to 1. For each $1\leq i \leq d$, the following $N_i+1$ linear equations holds
\begin{equation}
\left\lbrace \begin{array}{lclcll}
\omega_i^{(1)} & + & \dots & + & \omega_{i}^{(N_i+1)} & = 1 \\
\omega_i^{(1)} x_i^{(1)} & + & \dots & + & \omega_{i}^{(N_i+1)} x_{i}^{(N_i+1)} & = c_i^{(1)} \\
\multicolumn{1}{c}{\vdots}  & & & & \multicolumn{1}{c}{\vdots} & \multicolumn{1}{c}{\vdots \quad}\\
\omega_i^{(1)} {x_i^{(1)}}^{N_i} & + & \dots & + & \omega_{i}^{(N_i+1)} {x_{i}^{(N_i+1)}}^{N_i} & = c_{i}^{(N_i)} 
\end{array} \right.
\label{eq : Vandermonde weight system}
\end{equation} 
The determinant of the previous system is a Vandermonde matrix. Hence, the system is invertible as long as the $(x_i^{(j)})_j$ are distinct. 

\subsection{Step 4. Computation of the objective function}
\label{subsec: Step 4}
Thanks to Sections \ref{subsec: Step 1}, \ref{subsec: Step 2}, and \ref{subsec: Step 3}, we can compute the positions $(x_i^{(j)})_{1\leq j \leq N_i+1}$ and the weights $(\omega_i^{(j)})_{1\leq j \leq N_i+1}$ of the $i$th input of some $\mu \in\mathcal{A}_{\Delta}$. We can therefore compute the following probability of failure (in Equation (\ref{eq: Probability of faillure sum of weights})):
\begin{eqnarray*}
	 \mathbb{P}_{\mu} (G(X)\leq h) & = & \sum_{i_1=1}^{N_1+1} \dots \sum_{i_d=1}^{N_d+1} \omega_{1}^{(i_1)} \dots \omega_{d}^{(i_d)}\ \mathbbm{1}_{\{G(x_{1}^{(i_1)}, \dots, x_{d}^{(i_d)}) \leq h\}} \ ,
\end{eqnarray*}
We recall that the positions and consequently the weights, were determined using Corollary \ref{Th : Generation of measure} from a sequence of canonical moments $(p_k)_{N_i+1\leq k \leq 2N_i+1}\allowbreak\in\Gamma$. So that, the exploration of $\mathcal{A}_{\Delta}$ is parameterized with canonical moments in $\Gamma$. No constraints need to be enforced, as a discrete measure generated from the canonical moments naturally satisfies the moment constraints. The \textsc{p.o.f} is then optimized globally using a differential evolution solver \citep{price_differential_2005}. Algorithm \ref{Alg: p.o.f} summarizes step 1 to step 4 in order to compute the lowest probability of failure (\ref{eq: Probability of faillure sum of weights}). The main cost of the algorithm arises from the high number of metamodel calls for $G$, evaluated on a $d$-dimensional grid of size $\prod_{i=1}^d (N_i +1)$.

\begin{algorithm}[ht]
  \SetKwInOut{Input}{Input}

  \Input{- lower bounds, $\mathbf{l} = (l_1, \dots, l_d)$\\
	 - upper bounds, $\mathbf{u} = (u_1, \dots, u_d)$\\
   - constraints sequences of moments, $\mathbf{c}_i = (c_{i}^{(1)}, \dots, c_{i}^{(N_i)})$ and its corresponding sequences of canonical\\
    moments, $\mathbf{p}_i = (p_{i}^{(1)}, \dots, p_{i}^{(N_i)})$ for $1\leq i \leq d$.\vspace{0.5cm}}
    \SetKw{Return}{return}
    \SetKwProg{Fn}{function}{}{end function}
    \Fn{\textsc{P.O.F}($p_{1}^{(N_1+1)}, \dots, p_{1}^{(2N_1+1)} ,\dots, p_{d}^{(N_d+1)} , \dots, p_{d}^{(2N_d+1)}$)}{
      \For{$i=1,\dots, d$}{
        \For{$k=1,\dots, N_i$}{
          $P_{i\ast}^{(k+1)}=(X-l_i-(u_i-l_i)(\zeta_{i}^{2k} +\zeta_{i}^{(2k+1)}))P_{i\ast}^{(k)}-(u_i-l_i)^2 \zeta_{i}^{(2k-1)} \zeta_{i}^{(2k)} P_{i\ast}^{(k-1)}$\;
          }
          $x_{i}^{(1)} , \dots, x_{i}^{(N_i+1)} = \text{roots}(P_{i}^{*(N_i+1)})$\;
          $\omega_{i}^{(1)}, \dots, \omega_{1}^{(N_i+1)} = \text{weight}(x_{i}^{(1)}, \dots, x_{1}^{(N_i+1)} , \mathbf{c}_i)$ \;
        }
    
    \Return $\sum_{i_1=1}^{N_1+1} \dots \sum_{i_d=1}^{N_d+1} \omega_{1}^{(i_1)} \dots \omega_{d}^{(i_d)} \ \mathbbm{1}_{\{G(x_{1}^{(i_1)} , \dots, x_{d}^{(i_d)}) \leq h\}}$\;}
  \caption{Calculation of the \textsc{p.o.f}}
  \label{Alg: p.o.f}
\end{algorithm}

\section{Modified algorithm for inequality constraints}
\label{sec: Modified algorithm}
In the following, we consider inequality constraints for the moments. The optimization set reads
\[\mathcal{A}=\left\lbrace \mu = \otimes \mu_i \in \bigotimes_{i=1}^{d} \mathcal{P}([l_i,u_i])\; | \; \alpha_{i}^{(j)} \leq \mathbb{E}_{\mu_i}[x^{j}] \leq \beta_{i}^{(j)}\ , \ 1\leq j\leq N_i \right\rbrace\ . \]
One can notice that $\alpha_{i}^{(j)} \leq \mathbb{E}_{\mu_i}[x^{j}] \leq \beta_{i}^{(j)}$ is equivalent to enforcing two constraints, thus drastically increasing the dimension of the problem. However, it is possible to restrict ourselves to one constraint. Considering the convex function $\varphi_{i}^{(j)} : x \mapsto (x^{j}-\alpha_{i}^{(j)})(x^{j}-\beta_{i}^{(j)})$, Jensen's inequality states that $\varphi_{i}^{(j)}(\mathbb{E}_{\mu_i}(x)) \leq \mathbb{E}_{\mu_i}(\varphi_{i}^{(j)}(x))$. Therefore, the sole constraint $\mathbb{E}(\varphi_{i}^{(j)}(x)) \leq 0$ ensures $\alpha_{i}^{(j)} \leq \mathbb{E}_{\mu_i}[x^{j}] \leq \beta_{i}^{(j)}$. Without loss of generality we still consider measures $\mu_i$ that are convex combinations of $N_i+1$ Dirac masses, for $i=1, \dots, d$.

We now propose a modified version of Algorithm \ref{Alg: p.o.f} to solve the problem with inequality constraints. For $i = 1, \dots, d$, we denote the moments lower bounds $\boldsymbol{\alpha}_i = (\alpha_{i}^{(1)}, \dots, \alpha_{i}^{(N_i)})$ and the moments upper bounds $\boldsymbol{\beta}_i = (\beta_{i}^{(1)} , \dots, \beta_{i}^{(N_i)})$. We use Equation \eqref{eq: Moment's affine transformation} to calculate the corresponding moment sequence $\boldsymbol{\alpha}'_i$ and $\boldsymbol{\beta}'_i$ after affine transformation to $[0,1]$.

\begin{algorithm}[ht]
  \SetKwInOut{Input}{Input}

  \Input{- lower bounds, $\mathbf{l} = (l_1, \dots, l_d)$\\
	 - upper bounds, $\mathbf{u} = (u_1, \dots, u_d)$\\
   - moments lower bounds, $\boldsymbol{\alpha}'_i = ({\alpha'_{i}}^{(1)} , \dots, {\alpha'_{i}}^{(N_i)})$ for $i=1,\dots,d$ \\ 
   - moments upper bounds, $\boldsymbol{\beta}'_i = ({\beta'_{i}}^{(1)} , \dots, {\beta'_{i}}^{(N_i)})$ for $i=1,\dots,d$\vspace{0.5cm}}
    \SetKw{Return}{return}
    \SetKwProg{Fn}{function}{}{end function}
    \Fn{\textsc{P.O.F} (${c'_{1}}^{(1)} , \dots, {c'_{1}}^{(N_1)}, p_{1}^{(N_1+1)} , \dots, p_{1}^{(2N_1+1)} ,\dots, {c'_{d}}^{(1)},\dots, {c'_{d}}^{(N_d)}, p_{d}^{(N_d+1)}, \dots, p_{d}^{(2N_d+1)}$)}{
      \For{$i=1,\dots, d$}{
        \For{$k=1,\dots, N_i$}{$p_{i}^{(k)} = f({c'_{i}}^{(1)} , \dots {c'_{i}}^{(k)}) $ \tcc*{f transform moments to canonical moments}}
        \For{$k=1,\dots, N_i$}{
          $P_{i\ast}^{(k+1)} = (X-l_i-(u_i-l_i)(\zeta_{i}^{(2k)} +\zeta_{i}^{(2k+1)}))P_{i\ast}^{(k)}-(u_i-l_i)^2 \zeta_{i}^{(2k-1)} \zeta_{i}^{(2k)} P_{i\ast}^{(k-1)}$\;
          }
          $x_{i}^{(1)} , \dots, x_{i}^{(N_i+1)} = \text{roots}(P_{i}^{*(N_i+1)})$\;
          $\omega_{i}^{(1)}, \dots, \omega_{i}^{(N_i+1)} = \text{weight}(x_{i}^{(1)}, \dots, x_{i}^{(N_i+1)} , \mathbf{c}_i)$ \;
        }
    
    \Return $\sum_{i_1=1}^{N_1+1} \dots \sum_{i_d=1}^{N_d+1} \omega_{1}^{(i_1)} \dots \omega_{d}^{(i_d)} \ \mathbbm{1}_{\{G(x_{1}^{(i_1)} , \dots, x_{d}^{(i_d)}) \leq h\}}$\;}
  \caption{Calculation of the \textsc{p.o.f} with inequality constraints}
  \label{Alg : p.o.f inequality}
\end{algorithm}

The \textsc{p.o.f} of algorithm \ref{Alg : p.o.f inequality} has $d + 2\times \sum_{i=1}^d N_i$ arguments. The new parameters are actually the first ${(N_i)}_{|i=1,\dots,d}$th moments of the inputs that were previously fixed. A new step in the algorithm is needed to calculate the canonical moments up to degree $N_i$ for $i=1, \dots, d$. This ensures that the constraints are satisfied while the canonical moments from degree $N_i+1$ up to degree $2N_i+1$ can vary between $]0,1[$ in order to generate all possible measures. The increase of the dimension does not affect the computational times neither the complexity. Indeed, the main cost still arises from the large number of evaluation of the code $G$, that remains equal to $\prod_{i=1}^d (N_i+1)$. Once again this new \textsc{p.o.f} function can be optimized using any global solver. 

\section{Numerical tests on a toy example} 
\label{sec: Numerical tests on toy example}
\subsection{Presentation of the hydraulic model}
\label{subsec:hydraulic model}
In the following, we address a simplified hydraulic model \citep{pasanisi_estimation_2012}. This code calculates the water height $H$ of a river subject to a flood event. It takes four inputs whose initial joint distribution is detailed in Table \ref{tab:initial distribution hydraulic model}. It is always possible to calculate the \textit{plug-in} quantiles for those particular distributions. The result is given in Figure \ref{fig : Moment Order Constraints}, which present the initial CDF. However, as we desire to evaluate the robust quantile over a class of measures, we present in Table \ref{tab : Constraints for hydraulic model} the corresponding moment constraints that the variables must satisfy. The constraints are calculated based on the initial distributions, while the bounds are chosen in order to match the initial distributions most representative values.

\begin{table}[htb]
	\centering	\caption{Initial distribution of the 4 inputs of the hydraulic model.}
	\label{tab:initial distribution hydraulic model}
	\begin{tabular}{lrr}
		\toprule
		Variable & Description & Distribution \\ 
		\cmidrule(lr){3-3}\cmidrule(lr){2-2}
		$Q$ & annual maximum flow rate & $Gumbel(mode=1013, scale=558)$ \\
		$K_s$ & Manning-Strickler coefficient & $\mathcal{N}(\overline{x}=30,\sigma=7.5)$ \\
		$Z_v$ & Depth measure of the river downstream & $\mathcal{U}(49,51)$ \\
		$Z_m$ & Depth measure of the river upstream & $\mathcal{U}(54,55)$  \\
		\bottomrule
	\end{tabular}

\end{table}
\begin{table}[ht]
	\centering\caption{Corresponding moment constraints of the 4 inputs of the hydraulic model.}
	\label{tab : Constraints for hydraulic model}
	\begin{tabular}{lrrrr}
		\toprule
		Variable & Bounds & Mean & \specialcell{Second order \\ moment} & \specialcell{Third order \\ moment} \\ 
		\cmidrule(lr){2-2}\cmidrule(lr){3-3}\cmidrule(lr){4-4}\cmidrule(lr){5-5}%
		$Q$ & $[160, 3580]$ & $1320.42$ & $2.1632 \times 10^6$ & $4.18\times 10^9$ \\
		$K_s$ & $[12.55, 47.45]$ & $30$ & $949$ & $31422$ \\
		$Z_v$ & $[49,51]$ & $50$ & $2500$ & $125050$ \\
		$Z_m$ & $[54,55]$ & $54.5$ & $2970$ & $161892$ \\
		\bottomrule
	\end{tabular}

\end{table}
\noindent The height of the river $H$ is calculated through the analytical model
\begin{equation}
	H = \left(\frac{Q}{300K_s \sqrt{\frac{Z_m - Z_v}{5000}}} \right)^{3/5} \ .
	\label{eq:Hydraulic Model}
\end{equation}
We are interested in the flood probability $\sup_{\mu\in\mathcal{A}} \mathbb{P}(H \geq h)$.

\subsection{Maximum constraints order influence}
\begin{figure}[htb]
	\centering
	{\includegraphics[scale=0.35]{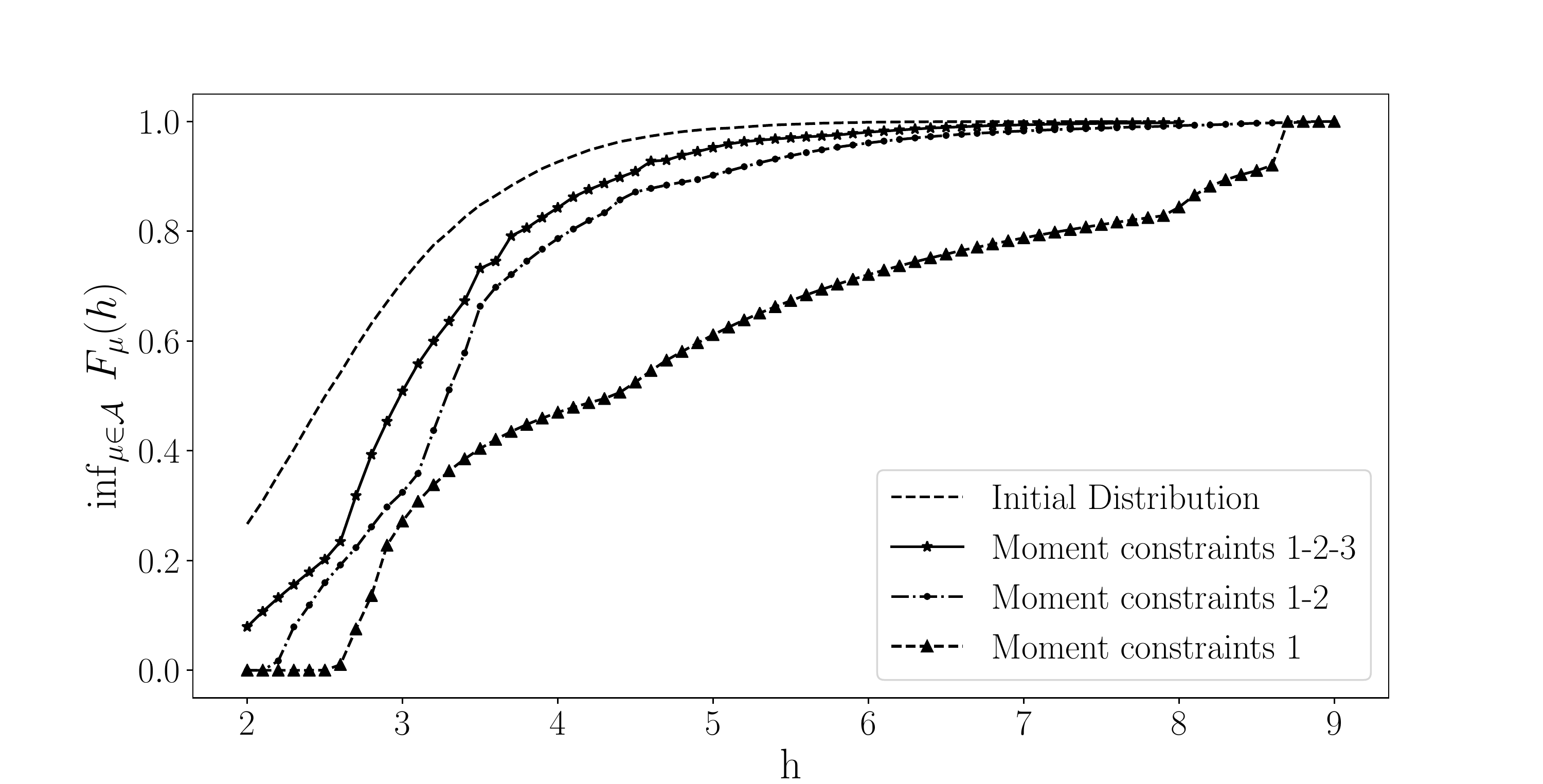}}
	\caption{$h$ is the height of the river. We observe the influence of the number of moment constraints enforced on the CDF lower envelope.} 
	\label{fig : Moment Order Constraints} 
\end{figure}
We compare the influence of the constraint order on the optimum. The initial distributions and the constraints enforced are available in Table \ref{tab : Constraints for hydraulic model}. The value of the constraints correspond to the moments of the initial distributions. Figure \ref{fig : Moment Order Constraints} shows how the size of the optimization space $\mathcal{A}$ decreases by adding new constraints. A differential evolution solver was used to perform the optimization. The initial CDF was computed with a Monte Carlo algorithm. One can observe that enforcing only one constraint on the mean will give a robust quantile significantly larger than the one of the initial distribution. On the other hand, adding three constraints on every inputs reduces quite drastically the space so that the optimal quantile found are close to the one of the initial CDF. A good compromise is to enforce two constraints on the first two moments. This is equivalent to enforcing the mean and the variance on the input distributions, which is a natural way of proceeding. Using this rule, the moment space is not to large so that the worst case quantile remains physically viable.

\subsection{Comparison with the Mystic framework}
We highlight the interest of the canonical moments parameterization by comparing its performances with the Mystic framework \citep{mckerns_optimal_2012}. Mystic is a Python toolbox suitable for OUQ. In Figure \ref{fig : Comparison Canonical Mystic HydraulicModel} one can see the comparison between Mystic and our algorithm. Both computations were realized with an identical solver, and computational times were similar ($\approx$30 min). We enforced one constraint on the mean of each input (see Table \ref{tab : Constraints for hydraulic model}). The performance of the Mystic framework is outperformed by our algorithm. Indeed, the generation of the weights and support points of the input distributions is not optimized in the Mystic framework. Hence, an intermediary transformation of the measure is needed in order to respect the constraints. During this transformation, the support points can be send out of bounds so that the measure is no more admissible. Many population vectors are rejected, which reduces the overall performance of the algorithm. Meanwhile, our algorithm warrants the exploration of the whole admissible set of measure without any vector rejection. 
\begin{figure}[htb]
	\centering
	{\includegraphics[scale=0.35]{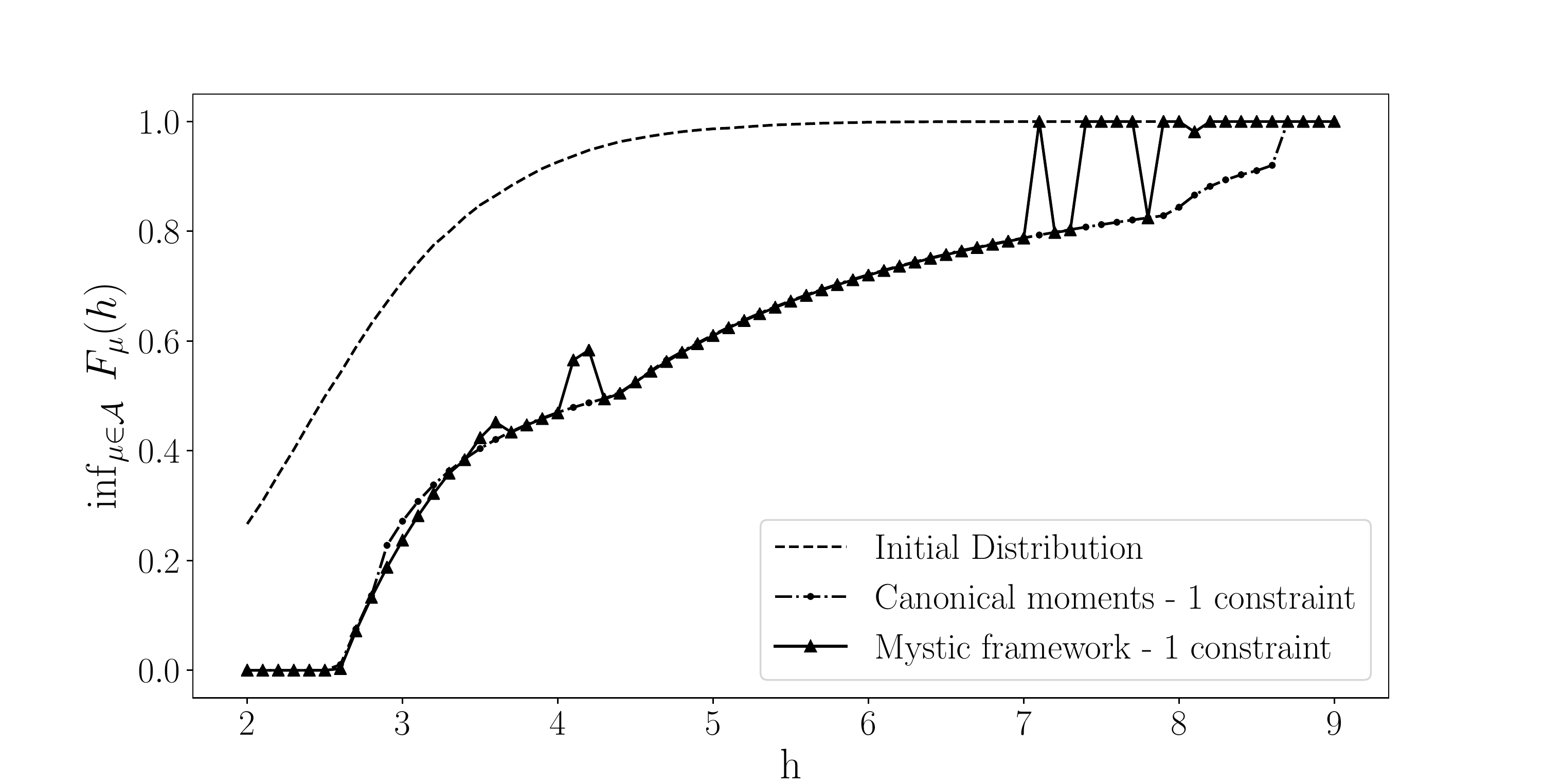}}
	\caption{Comparison of the performance of the Mystic framework and our algorithm on the hydraulic code.} 
	\label{fig : Comparison Canonical Mystic HydraulicModel} 
\end{figure}

\section{Application to the use-case}
\label{sec: Real case study}
\subsection{Presentation of the use-case}
Our use-case consists in thermal-hydraulic computer experiments, typically used in support of regulatory work and nuclear power plant design and operation. The numerical model is based on code CATHARE 2 (V2.5$\_$3mod3.1) which simulates the time evolution of physical quantities during a thermal hydraulic transient. The simulated accidental transient is an Intermediate Break Loss Of Coolant Accident (IBLOCA) with a break on the cold leg and no safety injection on the broken leg. One run of the computer code takes around 20 minutes on an ordinary computer. Nevertheless, in regards of both physical phenomena and dimensions of the system, we could consider our case under study as a simplified model with respect to a realistic modeling of a reactor. In this use-case, $d=27$ scalar inputs variables of CATHARE code are uncertain and defined by their probability density function. They correspond to various physical parameters, for instance: interfacial friction, critical flow rates, heat transfer coefficients, etc. They are all considered mutually independent. The output variable $Y$ is a single scalar which is the maximal peak cladding temperature (PCT) during the accident transient, see an example in Figure \ref{fig: RAW CATHARE}.
\begin{figure}[htb]
	\begin{minipage}[t]{0.48\linewidth}
		\centering
		\includegraphics[scale=0.18]{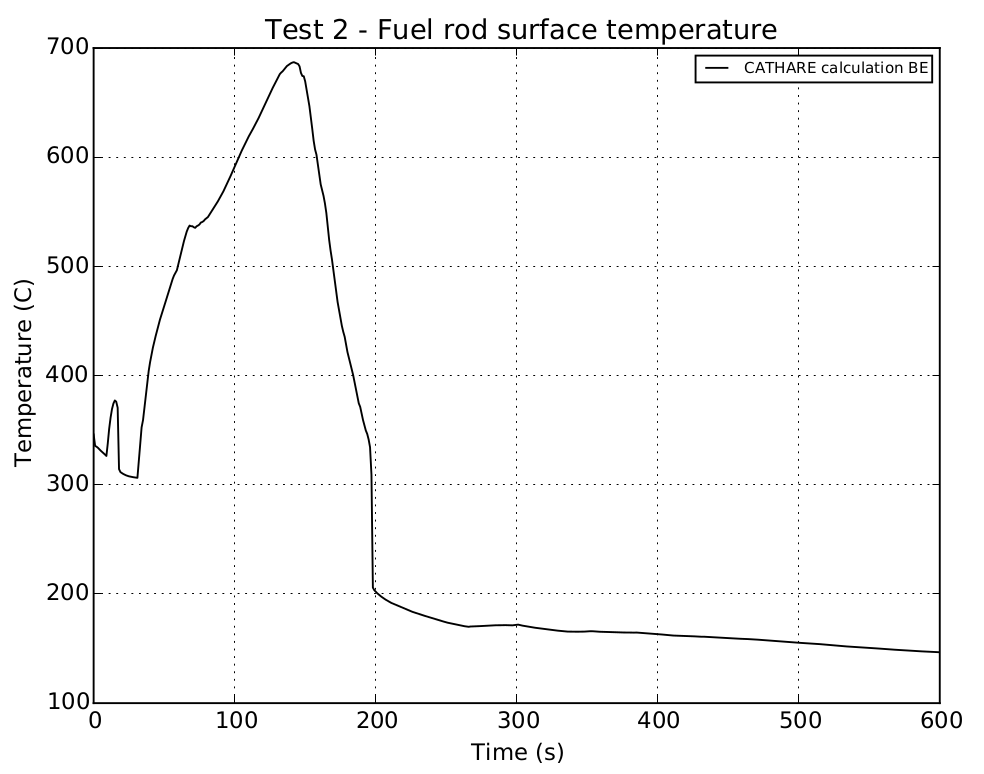}
		\caption{CATHARE temperature output for nominal parameters, the maximal temperature value is 687$^\circ$C after 140 secondes.}
		\label{fig: RAW CATHARE}
	\end{minipage}
	\hfill
	\begin{minipage}[t]{0.48\linewidth}
		\centering
		\includegraphics[scale=0.31]{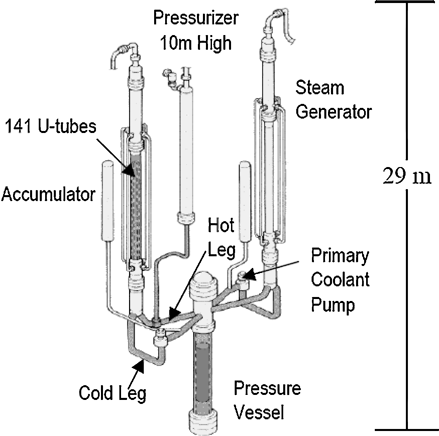}
		\caption{The replica of a water pressured reactor, with the hot and cold leg.}
		\label{fig: REP maquette}
	\end{minipage}
\end{figure}  

The number $n$ of simulations chosen for the design of experiments is a compromise between the CPU time required for each simulation and the number of input parameters. For uncertainty propagation and metamodel building purpose, it is a common rule to chose $n$ at least 10 times the dimension $d$ of the input vector \cite{loeppky_choosing_2009}. Here $n=1000$ simulations were performed using a space filling Latin Hypercube Sample (LHS) in dimension 27, thus providing a nice coverage of the high-dimensional input space \citep{fang_design_2005}.

A screening based on the Hilbert-Schmidt Independence Criterion (HSIC) dependence measure \citep{da_veiga_global_2013}, was performed on the $n=1000$ learning simulations. The hypothesis ``$\mathcal{H}_0^{(k)}$: the input $X_k$ and the output $Y$ are independent" was rejected for 9 inputs with a significance level $\alpha=0.1$. Those inputs, designated as PII, are given by Table \ref{tab: Constraints for CATHARE model}. The screening based on the HSIC takes into account the whole output variability. However, in some pathological case a low influential input for global variability may have an important impact on a $p$-quantile. A perspective to solve this issue, is to consider the conditional and target HSIC \cite[p.23]{raguet2018target}, this will be the aim of a future work.

\begin{table}[ht]
\centering
\caption{Corresponding moment constraints of the 9 most influential inputs of the CATHARE model \cite{iooss_advanced_2018}.}	\label{tab: Constraints for CATHARE model}
  \begin{tabular}{lrrrr}
    \toprule
		Variable & Bounds & Mean & \specialcell{Second \\ moment} & \specialcell{Initial distribution \\ (truncated on bounds)} \\ 
		\cmidrule(lr){2-2}\cmidrule(lr){3-3}\cmidrule(lr){4-4}\cmidrule(lr){5-5}%
		$n^\circ 10$ & $[0.1, 10]$ & $1.33$ & $3.02$ & $LogNormal(0,0.76)$   \\
		$n^\circ 22$ & $[0, 12.8]$ & $6.4$ & $45.39$ & $Normal(6.4,4.27)$  \\
		$n^\circ 25$ & $[11.1, 16.57]$ & $13.83$ & $192.22$ & $Normal(13.79$  \\
		$n^\circ 2$ & $[-44.9, 63.5]$ & $9.3$ & $1065$ & $Uniform(-44.9,63.5)$ \\
		$n^\circ 12$ & $[0.1, 10]$ & $1.33$ & $3.02$ & $LogNormal(0,0.76)$ \\
		$n^\circ 9$ & $[0.1, 10]$ & $1.33$ & $3.02$ & $LogNormal(0,0.76)$ \\
		$n^\circ 14$ & $[0.235, 3.45]$ & $0.99$ & $1.19$ & $LogNormal(-0.1,0.45)$ \\
		$n^\circ 15$ & $[0.1, 3]$ & $0.64$ & $0.55$ & $LogNormal(-0.6,0.57)$ \\
		$n^\circ 13$ & $[0.1, 10]$ & $1.33$ & $3.02$ & $LogNormal(0,0.76)$ \\
		\bottomrule
	\end{tabular}
\end{table}

The input distributions described in Table \ref{tab: Constraints for CATHARE model} derive from the CIRCE method \cite{damblin2019bayesian,decrecy_2001}, where in a Bayesian setting the posterior distributions are computed with respect to an experimental database. The moment presented in Table \ref{tab: Constraints for CATHARE model} correspond to the moments of the the calibrated distributions. Those are the moments enforced to the measures in the moment class. This way, we can compare the quantiles obtained through the \textit{plug-in} and \textit{full-Gp} approaches with the maximal quantile of the OUQ framework. Indeed, with these constraints the initial distribution is also an element of the moment class. 

The Gaussian process (Gp) is then build on the PII reduced space, conditioned from the available $n=1000$ learning simulations. We usually consider in computer experiments an anisotropic (stationary) covariance, and the covariance kernel is here chosen as a Mat\'ern 5/2. The metamodel accuracy is evaluated using the predictivity coefficient $Q^2$ \citep{gratiet_metamodel-based_2017}:
\[ Q^2 = 1- \frac{\sum_{i=1}^{n_{test}} (y_{i} - \widehat{y}_{i})^2}{\sum_{i=1}^{n_{test}} (y_{i} - \frac{1}{n_{test}}\sum_{i=1}^{n_{test}} y_{i} )^2} \]
where $(x_{i})_{1\leq i \leq n_{test}}$ is a test sample, $(y_{i})_{1\leq i \leq n_{test}}$ are the corresponding observed outputs and $(\widehat{y}_{i})_{1\leq i \leq n_{test}}$ are the metamodel predictions. We use a Leave one out strategy in order to perform the validation on the learning sample, we obtain $Q^2=0.92$. As already discussed in the introduction, once the predictive metamodel has been built, it can be used to perform uncertainty propagation and in particular, estimate quantiles.

\subsection{Results on the use-case}

Two constraints were enforced on the first two moments of each inputs as depicted in Table \ref{tab: Constraints for CATHARE model}. Indeed, informations on the input distributions are often reduced to the knowledge of the means and variances. Moreover, in regards to the toy example, only one constraints on the mean defines a moment class too large, which leads to a worst case scenario not physically viable. At last, more constraints increases the number of support points of each discrete measure on the extreme points, thus increases the computational cost. 

We successfully applied the methodology on the 9 dimensional restricted Gp metamodel of the CATHARE code. However, the computation was one day long for each threshold. We restricted the computation of the CDF to a small specific area of interest (high quantile 0.5-0.99) and we parallelized the task so that the computation did not exceed one week. One can compare, in Figure \ref{fig : Comparison Canonical Mystic CATHARE Model}, the results of the computation realized with the Mystic framework and our algorithm. Once again, the computations were performed with identical solvers. It confirms the difficulty for the Mystic framework to explore the whole space of admissible measures. Indeed, the maximal quantile obtained with the Mystic framework is lower than the result of our optimization. This proves the efficiency of the canonical moments parameterization to solve this optimization problem.
\begin{figure}[htb]
	\centering
	{\includegraphics[scale=0.35]{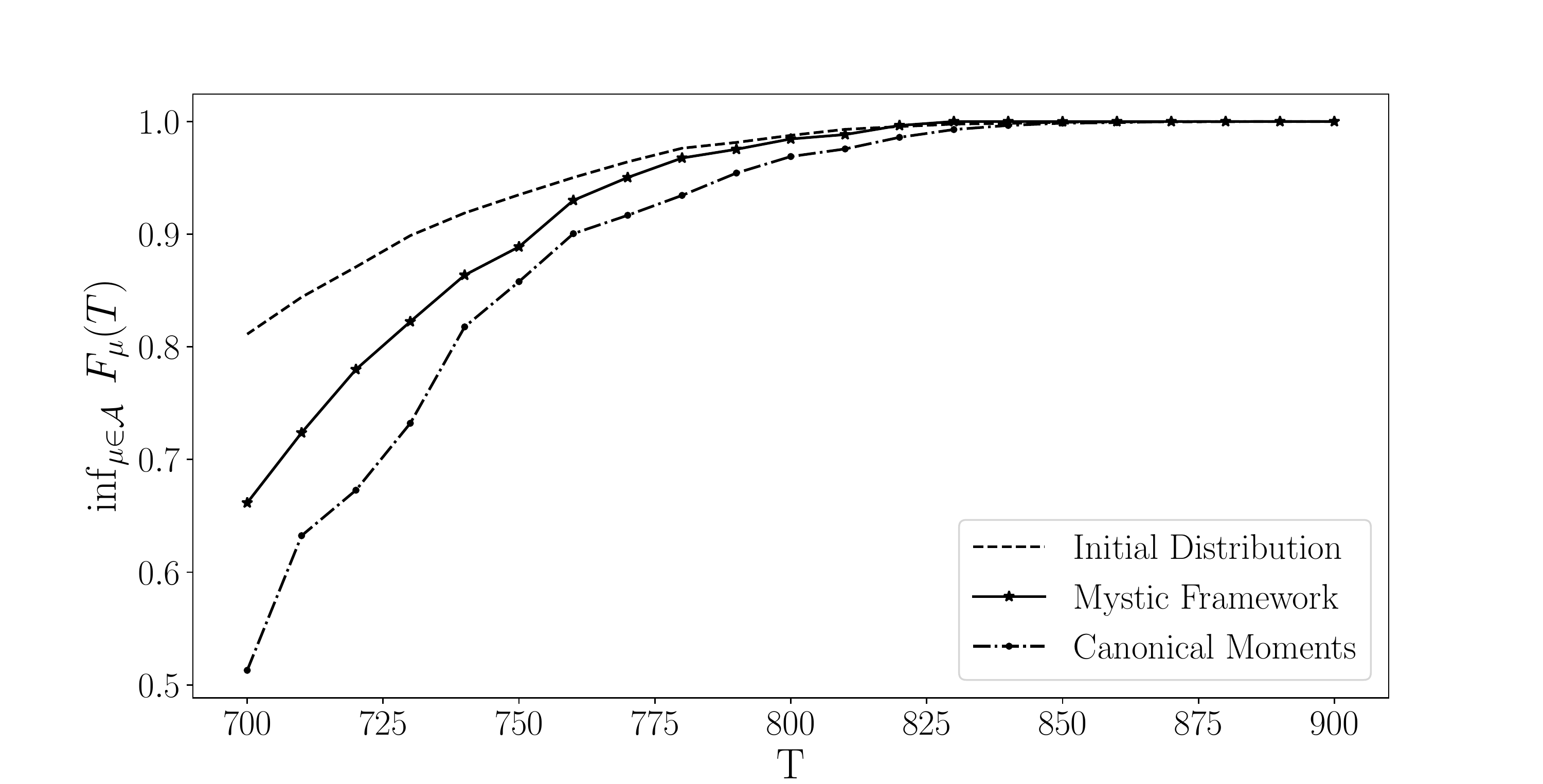}}
	\caption{Comparison of the performance and our algorithm on the 9 dimensional restricted CATHARE Gp metamodel. $h$ denotes the maximal peak cladding temperature during the accident transcient (see Figure \ref{fig: RAW CATHARE}).} 
	\label{fig : Comparison Canonical Mystic CATHARE Model} 
\end{figure}

One can compare the estimation of the 95$\%$-quantile of the peak cladding temperature for the IBLOCA application in Table \ref{tab: Quantile estimation}. A 90$\%$-confidence interval for the empirical quantile estimator was constructed with a bootstrap method. The \textit{plug-in} and \textit{full-Gp} approaches were defined in the introduction and correspond to a classical estimation of the quantile using respectively the predictor of the metamodel and the full Gaussian process \citep{iooss_advanced_2018}. The OUQ method corresponds to the maximal quantile, when the input distributions are only defined by two of their moments (see Table \ref{tab: Constraints for CATHARE model}). It is optimal in a sense that it minimize the uncertainty on the input distribution considering the available information. With this information, industrials are able to quantify the worst impact of the inputs uncertainty on the measure of risk, and adapt their choice of safety margins. 
\begin{table}[htb]
  \centering
	\caption{Results for the 95$\%$-quantile estimates.}
	\label{tab: Quantile estimation}
	\begin{tabular}{lrrrr}
		\toprule
		 & Empirical & Plug-in & Full-Gp & OUQ \\ 
     \cmidrule(lr){2-2}\cmidrule(lr){3-3}\cmidrule(lr){4-4}\cmidrule(lr){5-5}
		Mean & 746.80 & 735.83 & 741.46 & 788 \\
		90$\%$-CI & $[736.7, 747.41]$ & & $[738.76, 744.17]$ & \\
		\bottomrule
	\end{tabular}
\end{table}
The overall optimization can also be realized via Algorithm \ref{Alg : p.o.f inequality}, by relaxing the equality constraints into inequality constraints. This way, it is possible to quantify the sensitivity of the maximal quantile to the moment values.
\section{Summary}
\label{sec: Conclusion}

Metamodels are widely used in industry to perform uncertainty propagation, in particular to evaluate measures of risk such as high quantiles. In this work, we successfully increased the robustness of the quantile evaluation by removing the main sources of uncertainties tainting the inputs of the computer code. We evaluated the maximum measure of risk over a class of distribution. We focus on set of measures only known by some of their moments, and adapted the theory of canonical moments into an improved methodology for solving OUQ problems. 
Our objective function has been parameterized with the canonical moments, which allows the natural integration of the constraints. The optimization can therefore be performed free of constraints, thus drastically increasing its efficiency. The restriction to moment constraints suits most of practical engineering cases. We also provide an algorithm to deal with inequality constraints on the moments, if an uncertainty lies in their values. Our algorithm shows very good performances and great adaptability to any constraints order. However, the optimization is subject to the curse of dimension and should be kept under 10 input parameters. 

The joint distribution of the optimum is a discrete measure. One can criticize that it hardly corresponds to a physical, real world, interpretation. In order to address this issue, we will search for new optimization sets whose extreme points are not discrete measure. The unimodal class found in the literature of robust Bayesian analysis or the $\varepsilon$-contamination class, might be of some interest in this situation. New measures of risk will also be explored, for instance, superquantiles \citep{rockafellar_random_2014}, and Bayesian estimates associated to a given utility or loss function \citep{berger_statistical_1985}, which are of particular industrial interest.


\appendix

\section{Proof of duality proposition \ref{THM : DUALITY THEOREM}}\label{app: proof}

\begin{proof}
	we denote by $\displaystyle a = \sup_{\mu \in \mathcal{A}} \bigg[ \inf \left\lbrace h \in \mathbb{R}\ ; \ F_{\mu}(h) \geq p \right\rbrace\bigg]\ $ and $\displaystyle\ b = \inf \left\lbrace h  \in \mathbb{R}\; | \; \inf_{\mu\in\mathcal{A}} F_{\mu}(h) \geq p \right\rbrace$. In order to prove $a=b$, we proceed in two step. First step, we have
	\begin{eqnarray*}
		 & & \text{for all } h \geq b \ ;\  \inf_{\mu\in \mathcal{A}} F_\mu(h)  \geq p\ ,\\
		\Leftrightarrow & & \text{for all } h \geq b \ \text{ and for all } \mu\in\mathcal{A}\ ;\  F_\mu(h)  \geq p\ ,\\
		\Leftrightarrow & & \text{for all } \mu \in \mathcal{A}\ \text{ and for all } h \geq b \ ;\   F_\mu(h)  \geq p\ ,\\
		\Rightarrow & & \text{for all } \mu \in \mathcal{A}\ ;\ \inf\{h \in \mathbb{R} \, | \,  F_\mu(h) \geq p\}  \leq b \ ,
	\end{eqnarray*}
	so that $b \geq a$. Second step, because $a$ is the sup of the quantiles,
	\begin{eqnarray*}
			& & \text{for all } h\geq a\ ;\ \text{ for all }\mu\in\mathcal{A}\ ;\   F_\mu(h)  \geq p\ ,\\
			\Rightarrow & &  \text{for all } h\geq a\ ;\   \inf_{\mu\in\mathcal{A}} F_\mu(h)  \geq p\ ,
	\end{eqnarray*}
	so that
	\[\inf\ \bigg[ h  \in \mathbb{R} \ | \ \inf_{\mu \in \mathcal{A}}  F_{\mu}(h) \geq p \bigg]   \leq a \ ,\]
	and $b\leq a$.
\end{proof}

\section{Basic properties of continuous fraction}\label{app: continous fraction}
\begin{lemma}
	A finite continued fraction is an expression of the form 
	\[ b_0+\frac{a_1}{b_1+\frac{a_2}{b_2+ \dots}} = b_0 + \cfrac{a_1}{b_1} + \cfrac{a_2}{b_2} + \dots + \cfrac{a_n}{b_n} = \frac{A_n}{B_n} \ .\]
	The quantities $A_n$ and $B_n$ are called the $n$th partial numerator and denominator. There are basic recursive relations for the quantities $A_n$ and $B_n$ given by
	\begin{eqnarray*}
	A_n & = & b_n A_{n-1} + a_n A_{n-2} \ , \\
	B_n & = & b_n B_{n-1} + a_n B_{n-2} \ ,
	\end{eqnarray*}
	for $n \geq 1$ with initial conditions 
	\begin{eqnarray*}
	A_{-1}=1 & \quad , \quad & A_0=b_0 \ , \\
	B_{-1}=0 & \quad , \quad & B_0=1 \ .
	\end{eqnarray*}
	\label{Lem : Continued Fraction}
	\end{lemma}









\bibliographystyle{IJ4UQ_Bibliography_Style}

\bibliography{Biblio}
\end{document}